\documentclass[reprint,amsmath,amssymb,aps,pra]{revtex4-2}

\usepackage{graphicx}
\usepackage{dcolumn}
\usepackage{bm}

\usepackage[colorlinks,citecolor=blue,linkcolor=blue]{hyperref}

\usepackage{amsthm}
\usepackage{amsmath}
\newtheorem{lemma}{Lemma}
\newtheorem{definition}{Definition}
\newtheorem{theorem}{Theorem}

\usepackage[ruled]{algorithm2e}
\usepackage{algpseudocode}
\usepackage{color}
\usepackage{tcolorbox}

\begin{document}

\preprint{APS/123-QED}

\title{Characterization, synthesis, and optimization of quantum circuits over  multiple-control \texorpdfstring{$\textit{Z}$}{}-rotation gates: A systematic study}

\author{Shihao Zhang}
\email{zhangshh63@mail.sysu.edu.cn}
\author{Junda Wu}
\author{Lvzhou Li}
\email{lilvzh@mail.sysu.edu.cn}

\affiliation{Institute of Quantum Computing and Computer Theory, School of Computer Science and
Engineering, Sun Yat-sen University, Guangzhou 510006, China
}


\date{\today}

\begin{abstract}
We conduct a systematic study of quantum circuits composed of multiple-control $Z$-rotation (MCZR) gates as primitives, since they are widely-used components in quantum algorithms and also have attracted much experimental interest in recent years. Herein, we establish a circuit-polynomial correspondence to characterize the functionality of  quantum circuits over the MCZR gate set with continuous parameters. An analytic method for exactly synthesizing such quantum circuit to implement any given diagonal unitary matrix with an optimal gate count is proposed, which also enables the circuit depth optimal for specific cases with pairs of complementary gates. Furthermore, we 
present a gate-exchange strategy together with a flexible iterative algorithm for effectively optimizing the depth of any MCZR circuit, which can also be applied to quantum circuits over any other commuting gate set.
 Besides the theoretical analysis, the practical performances of our circuit synthesis and optimization techniques are further  evaluated by numerical experiments 
on two typical  examples in quantum computing, including diagonal Hermitian operators and Quantum Approximate Optimization Algorithm (QAOA) circuits with tens of qubits, which  can demonstrate a reduction in circuit depth by 33.40\% and 15.55\% on average over relevant prior works, respectively. Therefore, our methods and results provide a pathway  for implementing quantum circuits and algorithms  on recently developed devices.  
\end{abstract}

\maketitle


\section{Introduction}

With the arrival of the noisy intermediate-scale quantum (NISQ) era \cite{preskill2018quantum}, 
 the synthesis and optimization of quantum gate circuits have become the crucial step towards harnessing the power of
quantum computing on  realistic devices  \cite{leymann2020bitter,bharti2022noisy}. While 
single-qubit rotation and two-qubit controlled-NOT (CNOT) gates have received long-term 
investigations as they constitute an elementary gate set capable of universal quantum computation
\cite{barenco1995elementary,li2022circuit}, the multiple-control rotation (MCR) gates defined to 
act on more qubits also attract a great deal of interest from both fundamental and practical aspects:
\begin{enumerate}
\item[$\bullet$] Theoretically, MCR 
operations often serve as important components in many quantum algorithms or quantum 
computing models, such as preparing quantum hypergraph states \cite{rossi2013quantum,lin2018multiple}, 
building a circuit-based quantum random access memory \cite{park2019circuit,de2020circuit}, 
participating in Shor's factoring algorithm \cite{vandersypen2001experimental} and different 
types of quantum search algorithms 
\cite{figgatt2017complete,yoder2014fixed,roy2022deterministic},  quantum walk 
\cite{qiang2016efficient}, fault-tolerant quantum computation 
\cite{yoder2016universal,chao2018fault}. Therefore, a good understanding of MCR circuits can 
facilitate the design and analysis of new quantum information processing schemes. In fact, MCR gates have been included as basic building blocks in some popular quantum computing software frameworks, such as Qiskit \cite{Qiskit} and PennyLane \cite{PennyLane}. 
\item[$\bullet$] Instead of performing  concatenated single- and two-qubit gates in conventional experiments \cite{martinez2016compiling,figgatt2017complete,mandviwalla2018implementing}, recent experimental
progress has also been made for direct implementations of MCR gates in 
a variety of physical systems, including ion traps \cite{monz2009realization}, neutral atoms \cite{levine2019parallel}, linear and nonlinear quantum optics \cite{mivcuda2013efficient,dong2018polarization,ru2021realization}, and superconducting circuit 
systems \cite{fedorov2012implementation,song2017continuous,Kim2022}. In particular, MCR gates have 
been used as $native$  quantum gates in practical experiments for demonstrating quantum 
algorithms \cite{2020PhysRevApplied.14.014072,hill2021realization} and quantum error correction \cite{reed2012realization}. Therefore, quantum circuits over suitable  MCR gates for benchmarking and exploiting these ongoing quantum hardware need to be  specifically considered. 
\end{enumerate}

To our knowledge, several notable works have investigated quantum circuit models at the level of MCR gates with various techniques and results. 
For example, discussions about the use of multiple-control Toffoli gates as basic building blocks in circuit synthesis were presented in early years, including the use of Reed-Muller Spectra \cite{maslov2007techniques}, Boolean satisfiability (SAT)  techniques \cite{grosse2009exact}, or NCV-$|v_1\rangle$ libraries \cite{sasanian2012realizing}. 
 Typically, in 2014 the  issue of decomposing diagonal Hermitian quantum gates into a set consisting of solely multiple-controlled Pauli $Z$ operators has been studied  \cite{houshmand2014decomposition} by introducing a binary representation of these gates. In 2016,
different circuit identities that can replace certain configurations of the multiple-control Toffoli gates with their simpler multiple-control relative-phase implementations were  reported   \cite{maslov2016advantages}, showing the optimized resource counts. Given these promising results, quantum circuits based on a wider range of multiple-control quantum gates and their applications  are worthy of more in-depth exploration as well.

In this paper, we develop a systematic characterization, synthesis and optimization of quantum circuits over multiple-control $Z$-rotation (MCZR) gates with continuous parameters,
each of which would apply a $Z$-axis rotation gate $R_Z(\theta )=diag\{1,{e}^{i\theta }\}$ with a real-valued $\theta$ to the target qubit only when all its  control qubits are set to 1. 
In fact, such quantum gates play a prominent role in quantum state generation \cite{rossi2013quantum,nakata2014generating,gachechiladze2019changing,banerjee2020quantum}, quantum circuit construction \cite{mottonen2004quantum,bergholm2005quantum,maslov2018use},  
and fault-tolerant quantum computation \cite{yoder2016universal,chao2018fault}. Accordingly, schemes aimed at realizing fast and high-fidelity special or general MCZR gates are constantly being proposed \cite{2018multimode,Barnes2017Fastmicro,2018Su_one_step,2020FastMultiqubit,glaser2022controlled,2022Su_efficient} as well as  experimentally demonstrated \cite{song2017continuous, levine2019parallel,2020PhysRevApplied.14.014072,reed2012realization,hill2021realization} in recent years. In 2017, one-step implementation of the two-qubit $CZ$, three-qubit $CCZ$, and four-qubit $CCCZ$ gates has been realized with an experimental  fidelity of about 0.94, 0.868, and 0.817, respectively, based on the continuous-variable
geometric phase in a
superconducting circuit \cite{song2017continuous}. In 2020,  a multimode superconducting processor circuit with all-to-all connectivity that can implement the near-perfect generalized $CCZ(\theta)$ gates with an arbitrary angle $\theta$ as the native three-qubit controlled operations was  presented~\cite{2020PhysRevApplied.14.014072}, and experimentally demonstrated three-qubit Grover’s search algorithm and the quantum Fourier transform. 
Hence, how to perform quantum computing tasks over such gates with a low gate count and circuit depth is of practical significance, motivating us to conduct a systematic study in this work. For a general consideration, the number of control qubits, the set of acting qubits and the angle parameters $\theta$ of MCZR gates are all  unrestricted. Our main contributions are as follows:
\begin{enumerate}
    \item[$\bullet$] In Section \ref{Characterization}, we put forward a convenient  polynomial representation
to describe the functionality of the MCZR circuits, indicating that any realizable unitary matrix must be a diagonal one (see Eq.~\eqref{MatrixQC}).
\item[$\bullet$]  In Section \ref{synthesis algorithm}, 
we analytically derive a circuit synthesis method that can provide an optimal gate-count for implementing any given diagonal unitary matrix, which also achieves an optimal circuit depth for cases consisting of well-defined pairs of complementary gates (see \textbf{Theorem}~\ref{pair-wise_synthesis} ).
\item[$\bullet$] In Section \ref{depth optimization}, we consider how to reduce the circuit depth of any given MCZR circuit by proposing a gate-exchange strategy (see \textbf{Lemma}~\ref{theorem_complementary}) together with a flexible iterative depth-optimization algorithm (see \textbf{Algorithm}~\ref{greedy_layer}), which can yield better optimization results at the cost of more execution time. 
\item[$\bullet$]  In Section \ref{applications}, we validate the performance of our synthesis and depth-optimization methods for MCZR circuits by experimental evaluations on two typical examples, including the diagonal Hermitian quantum operators and Quantum Approximate Optimization Algorithm (QAOA) circuits, both of which show improvements over previous results. For the former, our constructed circuits on average  can achieve a 33.40\% depth reduction over the prior work~\cite{houshmand2014decomposition} for the circuit size $n\in [2,12]$. For the latter, our optimized circuit depth ranges from 3.00 to 4.05 for $n\in [6,50]$, and on average can reduce the circuit depth up to 58.88\% over randomly selected circuits and 15.55\%
over the results from Ref.~\cite{alam2020efficient}, respectively. Notably, here we achieve a nearly-constant depth for moderate-size QAOA circuits on 3-regular graphs.

\end{enumerate}

We expect the methods and results of this paper would be beneficial to the study of implementing quantum circuits and algorithms on specific quantum systems, and some  further directions are   discussed in Section~\ref{Conclusion}.

\section{\label{give_notations}Notation}
For convenience, here we introduce some notations used throughout this
paper. We denote the set $\{a,a+1,a+2,\ldots ,b\}$ by $\left[ a,b \right ]$ with $a,b$ being integers and $a\le b$. When $a=1$, notation $\left[ a,b \right ]$ is simplified to $\left[ b \right ]$.    For a binary number $x$, we use $q(x)=bin2dec(x)$ to represent its corresponding decimal number. The symbols ${|| v ||}$ and ${|S|}$
indicate the Hamming weight of a binary string $v$ (i.e. the number of 1s in $v$) and the size of the set $S$ (i.e. the number of its elements), respectively.  For an $n$-bit string $v=v_1v_2\ldots v_n$, we denote the  set of positions of all `1' bits as $P_v=\{p_1,p_2,\cdots,p_{|| v ||}\}$ such that $v_{p_1}=v_{p_2}=\cdots =v_{p_{||v||}}=1$. We use $I_{m \times n}$ to denote the size $m \times n$ identity matrix, and the symbol $\circ$ is used to concatenate $m(m\geq 2)$ subcircuits $\{QC_1,QC_2,\ldots,QC_m\}$ to form a circuit $QC$ such that $QC={QC_1}\circ{QC_2}\circ \ldots \circ {QC_m}$.  

\section{\label{Characterization}Characterization of MCZR circuits}

To characterize the functionality of the MCZR circuit,  
 we first establish a useful circuit-polynomial correspondence and then illustrate its  unitary matrix representation.

The MCZR gate family for an  $n$-qubit quantum circuit can be denoted as $\{{C}^{(k)}Z({{\theta }_{c,t}}):c\subset [n], t\in [n],k=\left| c \right|\}$, with $c$ being the control set, $t$ being the target, and ${\theta }_{c,t}$ being a $Z$-rotation  angle parameter.
By definition, the action of a MCZR gate on each  computational basis state is
\begin{align}\label{trans1}
    &{C}^{(k)}Z({{\theta }_{c,t}}):\left| {{x}_{1}},{{x}_{2}},\ldots,{{x}_{n}} \right\rangle \nonumber \\ \mapsto &\exp ({i}{{\theta}_{c,t}}{{x}_{t}}\prod\nolimits_{j\in c}{{{x}_{j}}})\left|{{x}_{1}},{{x}_{2}},\ldots,{{x}_{n}} \right\rangle.
\end{align}
The global phase factor in Eq.~\eqref{trans1} indicates that 
the function of gate ${C}^{(k)}Z({{\theta }_{c,t}})$ remains unchanged 
under any permutation of $k$ control and one target qubits in the set $act=c\bigcup t$.
Therefore, we can simply denote each MCZR gate acting on all  qubits in a set $act\subseteq [n]$ as $G(act,{\theta }_{act})$ such that
\begin{align}\label{trans2}
    &G(act,{\theta }_{act}):\left| {{x}_{1}},{{x}_{2}},\ldots,{{x}_{n}} \right\rangle \nonumber \\ \mapsto &\exp ({i}{{\theta}_{act}}\prod\nolimits_{j\in act}{{{x}_{j}}})\left|{{x}_{1}},{{x}_{2}},\ldots,{{x}_{n}} \right\rangle.
\end{align}
In this way, any quantum circuit $QC$ consisting of $m$ MCZR gates $G(act_1,{\theta }_{act_1})$,
$G(act_2,{\theta }_{act_2})$,\ldots,  $G(act_m,{\theta }_{act_m})$ can transform each basis state as
\begin{align}\label{transQC}
    QC: &\left| {{x}_{1}},{{x}_{2}},\ldots,{{x}_{n}} \right\rangle \nonumber \\ \mapsto &\exp ({i}\cdot {p(x_1,x_2,\ldots,x_n)})\left|{{x}_{1}},{{x}_{2}},\ldots,{{x}_{n}} \right\rangle,
\end{align}
with 
\begin{equation}\label{Phase_Poly}
    p({{x}_{1}},{{x}_{2}},\ldots ,{{x}_{n}})=\sum\limits_{k=1}^{m}{{{\theta }_{ac{{t}_{k}}}}\left( \prod\nolimits_{j\in ac{{t}_{k}}}{{{x}_{j}}} \right)}
\end{equation}
being a $phase$  $polynomial$ associated with the circuit $QC$. That is to say, any given $n$-qubit MCZR circuit $QC$ corresponds to a unique phase polynomial with real coefficients and degree at most $n$.   

Now we turn to the unitary matrix representation of $n$-qubit MCZR circuits. Eq.~\eqref{trans2} reveals that each MCZR gate can be explicitly expressed as a  diagonal unitary matrix of size $2^n\times 2^n$ as
\begin{equation}\label{Diagonal}
    G(act,{\theta }_{act})= \sum\limits_{x\in {{\{0,1\}}^{n}}}
    {\exp ({i}{{\theta}_{act}}\prod\nolimits_{j\in act}{{{x}_{j}}})\left|x \right\rangle \left\langle  x \right|},
\end{equation}
with all its diagonal elements being 1 or $e^{i\theta_{act}}$.
Since all MCZR gates are diagonal and commutative, two or more MCZR gates that act on the same set of  qubits in a circuit can be merged into one by just adding their angle parameters. Without loss of generality
, in this paper we focus on  the non-trivial MCZR circuit $QC$ such that
 all the  constituent $m$  gates have distinct qubit set $act_k(k=1,2,\ldots,m)$, and its unique phase polynomial in Eq.~\eqref{Phase_Poly} exactly has degree $\max \{\left| ac{{t}_{k}} \right|:k=1,2,\ldots,m\}$ and $m$ terms with real coefficients being the angle parameters $\left\{ {{\theta }_{ac{{t}_{k}}}}:k=1,2,\ldots ,m \right\}$. Accordingly, the circuit $QC$ in Eq.~\eqref{transQC} would function as a diagonal unitary matrix as
\begin{equation}\label{MatrixQC}
    D(QC)= \sum\limits_{x\in {{\{0,1\}}^{n}}} {\exp ({i}\cdot {p(x)})\left|x \right\rangle \left\langle  x \right|}, 
\end{equation}
with the polynomial  $p(x=x_1,x_2,\ldots,x_n)$ defined in Eq.~\eqref{Phase_Poly}.
Obviously, two MCZR circuits over different gate sets would implement two distinct diagonal unitary matrices.
For clarity, we display an instance circuit with $n$=3 and its polynomial as well as unitary matrix representation in Fig.~\ref{fig1}.
 
\begin{figure}[htp]
    \centering
    \includegraphics[width=0.5\textwidth]{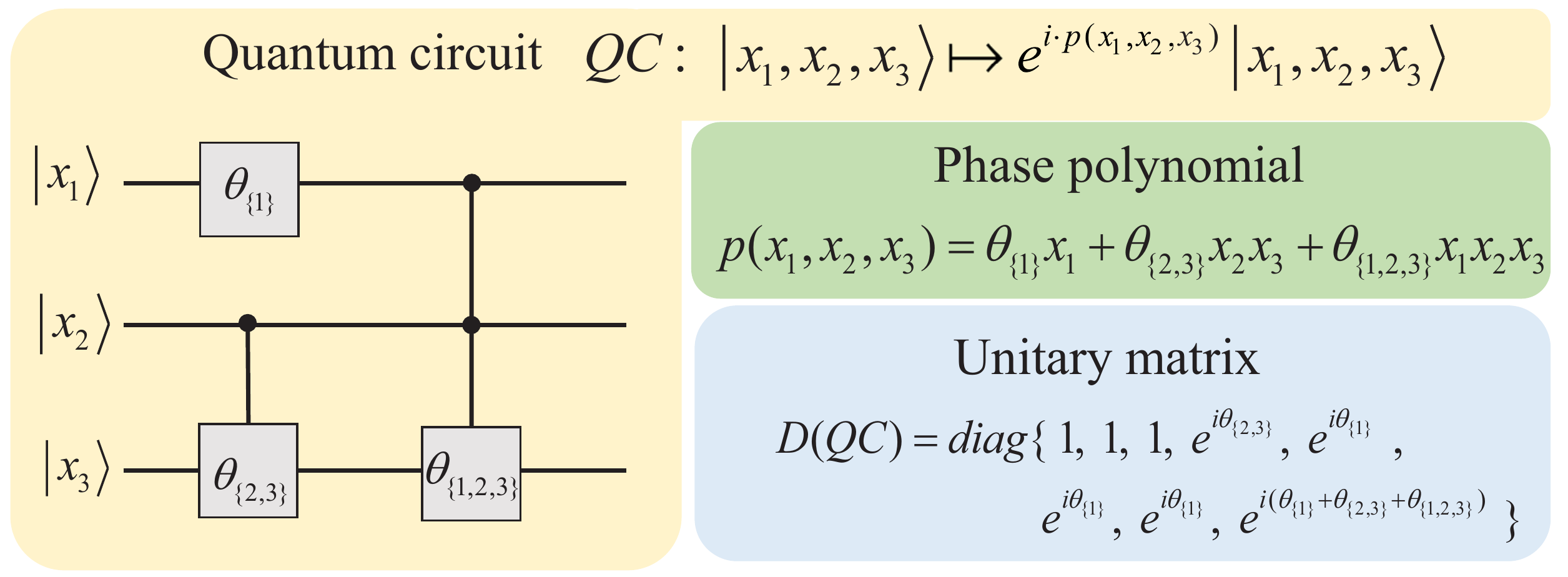}
    \caption{A three-qubit circuit $QC$ consisting of three MCZR gates with angle parameters  $\theta_{\{1\}}$, $\theta_{\{2,3\}}$, and $\theta_{\{1,2,3\}}$, respectively, which can add a phase factor $e^{i\cdot p(x_1,x_2,x_3)}$ to the basis state $\left| {{x}_{1}},{{x}_{2}},{{x}_{3}} \right\rangle$ with a phase polynomial $p(x_1,x_2,x_3)={\theta_{\{1\}}}{x_1}+\theta_{\{2,3\}}{x_2}{x_3}+\theta_{\{1,2,3\}}{x_1}{x_2}{x_3}$. The unitary matrix represented by $QC$ is a diagonal one as $D(QC)=diag\{1,1,1,e^{i\theta_{\{2,3\}}},e^{i\theta_{\{1\}}}, e^{i\theta_{\{1\}}}, e^{i\theta_{\{1\}}},  e^{i(\theta_{\{1\}}+\theta_{\{2,3\}}+\theta_{\{1,2,3\}})}    \}$.}
    \label{fig1}
\end{figure}

\section{Optimal synthesis of MCZR circuits}
\label{synthesis algorithm}

In above section, we have revealed that a MCZR circuit can implement a 
diagonal unitary matrix. 
This in turn raises a natural question: can an arbitrary diagonal operator be implemented by a MCZR gate circuit exactly? This is an attractive subject since diagonal unitary matrices
 have a wide range of applications in quantum computing and quantum information \cite{qiang2016efficient, 2008simulation,  nakata2014diagonal,  alam2020efficient, XU2023quantum}.
 
In this section, 
we address this issue  by proposing a circuit synthesis method to construct an $n$-qubit gate-count optimal MCZR circuit for implementing a size $N\times N\;(N=2^n)$ diagonal unitary matrix
\begin{align}\label{Given Diagonal}
  D(\overrightarrow{\alpha}=\left[\alpha_0,\alpha_1,\ldots,\alpha_{N-1} \right])&=\left[
    \begin{array}{ccccc}
      e^{i\alpha_0}&0&\cdots&0&0\\
      0&e^{i\alpha_1}&\cdots&0&0\\
      \vdots&\vdots&\ddots&\vdots&\vdots\\
      0&0&\cdots&0&e^{i\alpha_{N-1}}
    \end{array}
    \right]\nonumber\\
    &=\sum\limits_{x\in {{\{0,1\}}^{n}}} {\exp ({i} {\alpha_{q(x)}})\left|x \right\rangle \left\langle  x \right|}
\end{align}
with $q(x) =  bin2dec(x)$, which also enables the circuit depth optimal for specific cases. In particular, we emphasize that the $optimality$ mentioned in this paper always indicates an $exact$ optimal value for the gate count and circuit depth rather than  $asymptotically$ optimal results, indicating that our optimal results cannot be improved any more.
For convenience,  here we rewrite each available gate $G(act,{\theta
}_{act})$ in Eq.~\eqref{trans2} as $G(v,{\theta
}_{v})$ by associating $act$ with an $n$-bit string $v={v_1}{v_2}\ldots{v_n} \in \{0,1\}^n$ such that 
\begin{equation}\label{string_v}
    v_j :=
    \begin{cases}
     1, & j\in act;\\
       0, & j\in [n]\backslash act.
    \end{cases}
\end{equation}

Our main results in this section are summarized as \textbf{Theorems}~\ref{theorem_synthesis1}, \ref{theorem_optimal_depth}, and \ref{pair-wise_synthesis}.

\begin{theorem}
    \label{theorem_synthesis1}
    The MCZR gate set $\{G(v,\theta_v)\}$ for implementing a target diagonal unitary matrix $D(\overrightarrow{\alpha})$ in Eq.~\eqref{Given Diagonal} with $2^n$ given parameters $[\alpha_0,\alpha_1,\ldots,\alpha_{N-1} ]$ is unique, 
    and each gate parameter can be computed analytically as 
\begin{equation}
\label{theorem_angle}
\theta_v=(-1)^{|| v ||}\sum\limits_{x:P_x\subseteq P_v}{(-1)^{|| x ||}{\alpha_{q(x)}}},  \quad v\in \{0,1\}^n,
\end{equation}
with $q(x)$, $P_{v}(P_x)$, and $||v||(||x||)$ defined in Section \ref{give_notations}. Since $\theta_v$ indicates a trivial identity gate that can be omitted, the optimal gate-count for implementing $D(\overrightarrow{\alpha})$  is thus $|\{G(v,\theta_v\neq 0) \}|$ with $\theta_v$ from Eq.~\eqref{theorem_angle}.

\end{theorem}

\begin{proof}

According to Eq.~\eqref{string_v}, there are totally $2^n-$1  different types of  gates $\{G(v,{\theta }_v):v\in \{0,1\}^n\backslash 00..0 \}$ available to construct a MCZR circuit $QC$ that functions as  Eq.~\eqref{MatrixQC}, with its phase polynomial $p(x)$ in Eq.~\eqref{Phase_Poly}  rewritten as
\begin{equation}\label{Phase_Poly_2}
    p(x)=\sum\limits_{v\in \{0,1\}^n\backslash 00..0}{{{\theta }_{v}} ({x_1}^{v_1}{x_2}^{v_2}\ldots{x_n}^{v_n}) }.
\end{equation}

Since two quantum circuits which differ only by a global phase factor are equivalent, we suppose that a circuit $QC$ described by Eq.~\eqref{MatrixQC} can perform the target diagonal  matrix $D(\overrightarrow{\alpha})$ in Eq.~\eqref{Given Diagonal} as
\begin{align}
    \label{Realize_D}
    e^{i\theta_{00..0}}\sum\limits_{x\in {{\{0,1\}}^{n}}} &{\exp ({i}\cdot {p(x)})\left|x \right\rangle \left\langle  x \right|} \nonumber \\  &=\sum\limits_{x\in {{\{0,1\}}^{n}}} {\exp ({i} {\alpha_{q(x)}})\left|x \right\rangle \left\langle  x \right|},
\end{align}
leading to
\begin{equation}\label{Realize_equations}
\theta_{00..0}+p(x)=\alpha_{q(x)},\quad x\in \{0,1\}^n
\end{equation}
with $\theta_{00..0}$ being a global phase factor, $p(x)$ in Eq.~\eqref{Phase_Poly_2} and $q(x) =  bin2dec(x)$ defined in Section \ref{give_notations}. In total, Eq.~\eqref{Realize_equations} gives us $2^n$ linear equations  as
\begin{equation}
\label{linear_equa}
\begin{cases}
     & x=00..00: ~   \theta_{00..00}=\alpha_0;  \\  & x=00..01: ~   \theta_{00..00}+\theta_{00..01}=\alpha_1;  \\  & x=00..10: ~   \theta_{00..00}+\theta_{00..10}=\alpha_2;   \\ 
    & x=00..11: ~   \theta_{00..00}+\theta_{00..01}+\theta_{00..10}+\theta_{00..11}=\alpha_3;   \\
    & \vdots   \\  & x=11..11: \quad \sum\limits_{v\in {{\{0,1\}}^{n}}}{{{\theta }_{v}}}={{\alpha }_{N-1}}.
    \end{cases}
\end{equation}
Thus, if we can solve a set of $2^n$ angle parameters  $\{\theta_v:v\in \{0,1\}^n \}$ satisfying Eq.~\eqref{linear_equa} for any given $\overrightarrow{\alpha}=\left[\alpha_0,\alpha_1,\ldots,\alpha_{N-1} \right]$, then we obtain a MCZR circuit over the gate set $\{G(v,{\theta }_v)\}$ for implementing any $D(\overrightarrow{\alpha})$ in Eq.~\eqref{Given Diagonal}. In the following, we give an exact analytical expression of the solution to Eq.~\eqref{linear_equa} and prove its uniqueness. 

The linear equations in Eq.~\eqref{linear_equa}  can be succinctly summarized into
 a standard form as
\begin{equation}
\label{standard_form}    
J \cdot \left( \begin{matrix}
   {{\theta }_{00..00}}  \\
   {{\theta }_{00..01}}  \\{{\theta }_{00..10}}\\{{\theta }_{00..11}}\\
   \vdots   \\
   {{\theta }_{11..11}}  \\
\end{matrix} \right)=\left( \begin{matrix}
   {{\alpha }_{0}}  \\
   {{\alpha }_{1}}  \\   {{\alpha }_{2}}\\   {{\alpha }_{3}}\\  
   \vdots   \\
   {{\alpha }_{N-1}}  \\
\end{matrix} \right)
\end{equation}
such that the size $2^n\times 2^n$ coefficient matrix  $J$ has elements
\begin{equation}
\label{coefficient_matrix}
  J_{\widetilde{q}(x),\widetilde{q}(v)}=\begin{cases}
    1, & P_v\subseteq P_x ;\\
       0, & otherwise ,
    \end{cases} \quad x,v\in \{0,1\}^n,
\end{equation}
where the function $\widetilde{q}(\cdot )=bin2dec(\cdot)+1$ transforms a binary string into a decimal number as the row or  column index of a matrix, and the set $P_{x(v)}$ about a string $x(v)$ is defined in Section \ref{give_notations}. 
Consider another size $2^n\times 2^n$ matrix denoted $K$ with elements
\begin{equation}
\label{inverse_matrix}
  K_{\widetilde{q}(v),\widetilde{q}(x)}=\begin{cases}
    (-1)^{|| v ||+|| x ||}, & P_x\subseteq P_v ;\\
       0, & otherwise ,
    \end{cases} ~ x,v\in \{0,1\}^n,
\end{equation}
here we can prove the product of two matrices in Eqs.~\eqref{inverse_matrix} and \eqref{coefficient_matrix} as $Q=K\cdot J$ is exactly an identity matrix of size $2^n\times 2^n$. By definition, the matrix elements of  $Q$ are
\begin{align}
\label{product_matrix}
  Q_{\widetilde{q}(v_1),\widetilde{q}(v_2)}&=\sum\limits_{x\in {\{0,1\}}^{n}}{{{K}_{\widetilde{q}({{v}_{1}}),\widetilde{q}(x)}}{{J}_{\widetilde{q}(x),\widetilde{q}({{v}_{2}})}}} \nonumber\\
  &=(-1)^{||v_1||}\sum\limits_{x:{{P}_{{{v}_{2}}}}\subseteq {{P}_{x}} \subseteq {{P}_{v_1}}}{{(-1)}^{||x||}}+0, \nonumber \\ & v_1,v_2\in \{0,1\}^n.
\end{align}
 For the diagonal element of $Q$ with $v_1=v_2$ and $P_{v_1}=P_{v_2}$, Eq.~\eqref{product_matrix} turns into
\begin{equation}
\label{diagonal_Q}    Q_{\widetilde{q}(v_1),\widetilde{q}(v_1)}=(-1)^{||v_1||} \cdot (-1)^{||v_1||}=1, \quad v_1 \in \{0,1\}^n
\end{equation}
by taking $x=v_1$. For the off-diagonal elements of $Q$ with $v_1 \ne v_2$ and $P_{v_1} \ne P_{v_2}$, we have two cases: 
\begin{itemize}
\item[(i)] $P_{v_2} \not\subset P_{v_1}$, then no string $x$ can satisfy ${{P}_{v_2}}\subseteq {{P}_{x}} \subseteq {{P}_{v_1}}$, leading Eq.~\eqref{product_matrix} to
$Q_{\widetilde{q}(v_1),\widetilde{q}(v_2)}=0$; 
\item[(ii)] $P_{v_2} \subset P_{v_1}$, then there are totally $2^{||v_1||-||v_2||}$ strings $x$ that can satisfy ${{P}_{v_2}}\subseteq {{P}_{x}} \subseteq {{P}_{v_1}}$, wherein $||x||$ is even for exactly half of these $x$ and odd for the other half, leading Eq.~\eqref{product_matrix} to
$Q_{\widetilde{q}(v_1),\widetilde{q}(v_2)}=0$. 
\end{itemize}

At this point, we prove that $K \cdot J=I_{2^n \times 2^n}$ and thus the square matrix $K$ defined in Eq.~\eqref{inverse_matrix} is the unique inverse matrix of the coefficient matrix $J$ in  Eq.~\eqref{standard_form} by the common knowledge of linear algebra. By multiplying both sides of Eq.~\eqref{standard_form} with $K$ and using Eq.~\eqref{inverse_matrix}, we obtain an analytic form of the solutions $\{\theta_v\}$ to Eq.~\eqref{standard_form} as
\begin{align}
\label{Analytic_angle}
\theta_v &=\sum\limits_{x\in \{0,1\}^n}{{K_{\widetilde{q}(v),\widetilde{q}(x)}}\alpha_{q(x)}} \nonumber \\
&=(-1)^{|| v ||}\sum\limits_{x:P_x\subseteq P_v}{(-1)^{|| x ||}{\alpha_{q(x)}}},  \quad v\in \{0,1\}^n,
\end{align}
with $q(x)$, $P_{v}(P_x)$, and $||v||(||x||)$ defined in Section \ref{give_notations}.

\begin{figure*}[t]
    \centering
   \includegraphics[width=0.9\linewidth]{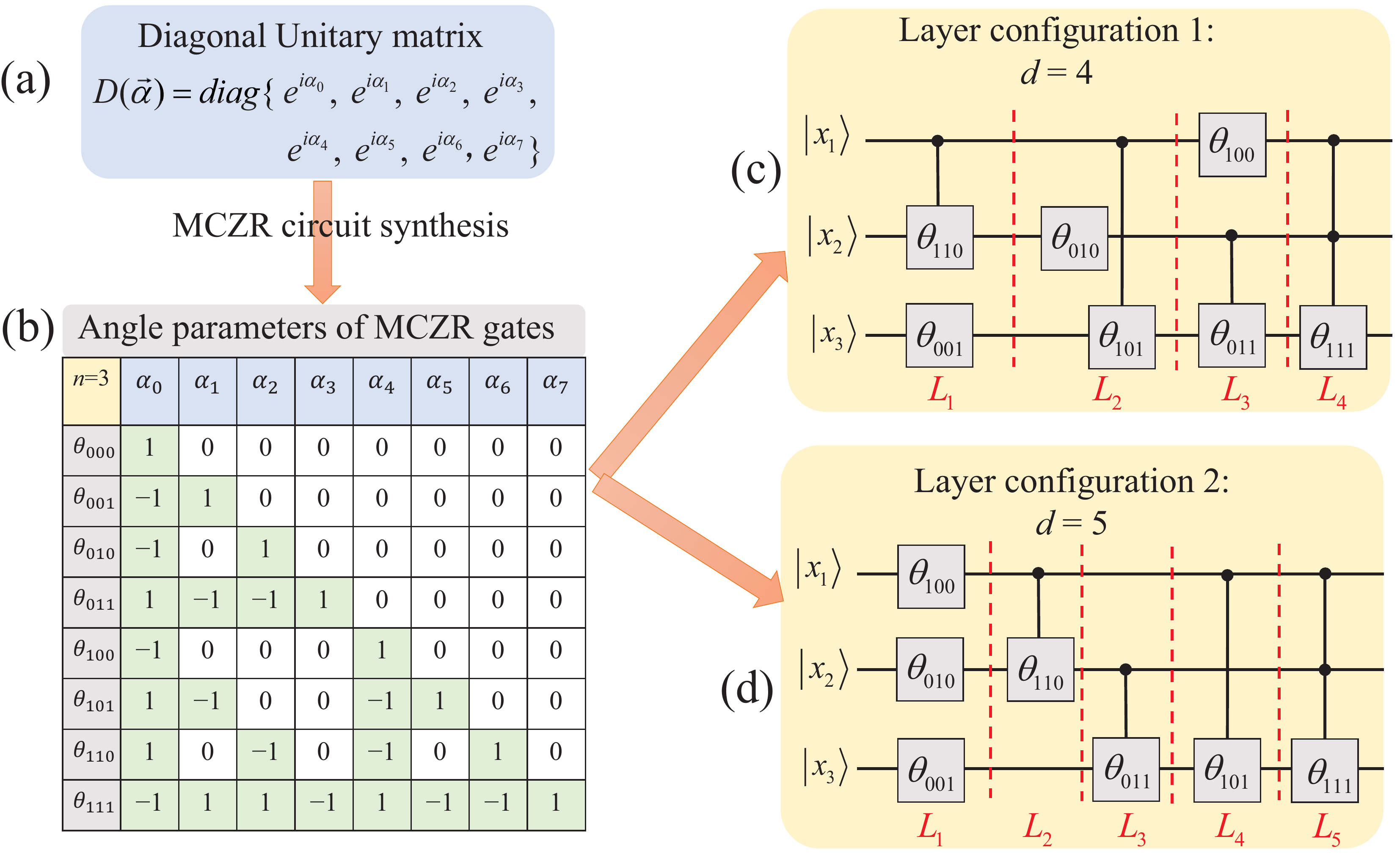}
    \caption{Example with $n=3$ to show the gate-count optimal synthesis of quantum MCZR circuits. To construct a circuit for realizing a given diagonal unitary matrix  $D(\protect \overrightarrow{\alpha})$   of size $8\times8$ in (a), we can first use Eq.~\eqref{theorem_angle} to  solve the angle parameters $\{\theta_v: v\in \{0,1\}^3\}$ of all employed  MCZR gates as linear combinations of given $\{\alpha_0, \alpha_1,\ldots,\alpha_7\}$ with non-zero coefficients marked green shown in (b).  Note the angle  parameter  $\theta_v=0$ indicates a trivial identity gate that can be  removed in the circuit. Then, these gates are arranged in different layers to give a circuit layer configuration.  For a general case, we present a circuit consisting of all gates in complementary pairs with a  depth $d=4$ in (c), while another circuit with  a depth $d=5$ is depicted in (d)  for comparison. As a summary, the circuit in (c) to implement (a) 
    can be directly obtained by \textbf{Theorem}~\ref{pair-wise_synthesis}.}
    \label{fig2}
\end{figure*}
 
  In summary,  Eq.~\eqref{Analytic_angle} represents a unique set of solutions so that the resultant MCZR circuit for implementing $D(\overrightarrow{\alpha})$ in Eq.~\eqref{Given Diagonal} naturally 
  achieves an optimal gate count. The angle parameter $\theta_v=0$ indicates its associated MCZR gate $G(v,\theta_v)$ is a trivial identity gate  that can be omitted. Therefore, the optimal gate count for realizing any diagonal unitary operator in Eq.~\eqref{Given Diagonal} is $\left| \{G(v,{{\theta }_{v}}\ne 0)\} \right|$ with the  gate parameters  obtained from Eq.~\eqref{Analytic_angle}, and in the worst case is $2^n-1$ when all 
  angle parameters are solved to be non-zero.    For clarity, an example with $n=3$ is shown in Figs.~\ref{fig2} (a) and (b).
\end{proof}

As a by-product, the uniqueness of the gate set $\{G(v,\theta_v)\}$ for implementing a diagonal unitary matrix as declared in  \textbf{Theorem}~\ref{theorem_synthesis1} gives us \textbf{Lemma}~\ref{indep_gate}.

\begin{lemma}
    \label{indep_gate}
    All MCZR gates in $\{G(v,\theta_v):v\in \{0,1\}^n,{\theta_v}\in [0,2\pi)  \} $ are independent, that is, none of them can be decomposed into a combination of the others. 
\end{lemma}

Besides the gate count, the circuit depth is 
    another important circuit cost metric that needs attention, since a reduced circuit-depth means less circuit execution time. A quantum circuit can be represented as a directed acyclic
graph (DAG) in which each node corresponds to a circuit's
gate and each edge corresponds to the input/output of a gate.
Then the circuit depth $d$  is defined as the maximum length
of a path flowing from an input of the circuit to an output \cite{amy2013meet}. Equivalently speaking, $d$ is the number of layers of quantum gates that compactly act on  disjoint sets of qubits   \cite{bravyi2018quantum, abdessaied2016reversible}. For example, the depth of the circuit in Fig.~\ref{fig1} with three non-zero angle parameters is $d=2$. 
Notice that a set of MCZR gates may form distinct layer configurations with respective circuit depths, as exemplified by the comparison between the depth-4 circuit in Fig.~\ref{fig2}(c) and depth-5 circuit in Fig.~\ref{fig2}(d). More generally, in \textbf{Theorem}~\ref{theorem_optimal_depth} we reveal the optimal circuit depth of any MCZR circuit constructed from   pairs of complementary gates as defined in \textbf{Definition}~\ref{gate-pair}.

\begin{definition}
    \label{gate-pair}
    We call a pair of MCZR gates $G(v_1,\theta_{v_1})$ and $G(v_2,\theta_{v_2})$ are complementary if and only if they  satisfy $v_1 \oplus v_2=11..11$.
\end{definition}

\begin{theorem}
    \label{theorem_optimal_depth}
    The optimal circuit depth of any MCZR circuit constructed from $d_1$ pairs of complementary gates is exactly $d_1$.
\end{theorem}

\begin{proof}
Suppose we construct an $n$-qubit MCZR circuit over $d_1$ pairs of complementary gates $\{G(v,\theta_v)\}$ by arranging them into $d$ layers denoted $\{L_1, L_2,\ldots,L_d\}$ such that all gates in each layer $L_i$ $(i=1,2,\ldots,d)$ are disjoint. Here we prove the minimum value of $d$ is $d_1$.

For brevity, we denote each gate layer $L_i$ by an  $n$-bit string as
\begin{equation}
\label{each_layer}
s(L_i)=\sum\limits_{v:G(v,{{\theta }_{v}})\in {{L}_{i}}}{v}, \quad i=1,2,\ldots,d,
\end{equation}
and all $d$ such strings totally own $nd$ bits of 0 and 1. On the other hand, the total number of `1' bits in 
$2{d_1}$ strings $v$  representing these gates is $n{d_1}$. Therefore, we have

\begin{equation}
nd\ge n{d_1}
\end{equation}
and the lower bound of circuit depth as
\begin{equation}
\label{optimal_depth}
d\ge {d_1}.    
\end{equation}

 Obviously, the equality in Eq.~\eqref{optimal_depth} can be achieved when every gate layer $L_i$ $(i=1,2,\ldots,d)$ has a pair of complementary gates, thus forming a circuit with an optimal depth $d_1$.
\end{proof}

A typical application of \textbf{Theorem}~\ref{theorem_optimal_depth} is to construct a depth-optimal MCZR circuit over  all $2^n-1$ non-zero gate parameters solved from \textbf{Theorem}~\ref{theorem_synthesis1} for implementing a given diagonal operator. That is, when all these gates  
are arranged into $(2^{n}-2)/2=2^{n-1}-1$ layers of complementary gates as $L_1=[v=00..01,v=11..10]$, $L_2=[v=00..10,v=11..01]$,\ldots, $L_{2^{n-1}-1}=[v=01..11,v=10..00]$  plus a sole gate in $L_{2^{n-1}}=[v=11..11]$, a circuit with an optimal depth $2^{n-1}$ is obtained. For clarity, a  circuit example with $n=3$ and the optimal depth $d=4$ is shown in Fig.~\ref{fig2}(c), while another circuit with a larger depth $d=5$ is shown in Fig.~\ref{fig2}(d) for comparison. 

Finally, the combination of \textbf{Theorem}~\ref{theorem_synthesis1} and \textbf{Theorem}~\ref{theorem_optimal_depth} leads to a pair-wise circuit synthesis method described as  \textbf{Theorem}~\ref{pair-wise_synthesis}.

\begin{theorem}[Pair-wise MCZR circuit synthesis]
    \label{pair-wise_synthesis}
    A MCZR circuit $QC$ over the gate set $\{G(v,\theta_v)\}$ for implementing an arbitrary diagonal unitary matrix $D(\overrightarrow{\alpha})$ in Eq.~\eqref{Given Diagonal} can be synthesized by computing each gate parameter $\theta_v$ according to Eq.~\eqref{theorem_angle} in a pair-wise way as $L_1=[v=00..01,v=11..10]$, $L_2=[v=00..10,v=11..01]$,\ldots, $L_{2^{n-1}-1}=[v=01..11,v=10..00]$, $L_{2^{n-1}}=[v=11..11]$ such that $QC={L_1}\circ{L_2}\circ \ldots \circ{L_{2^{n-1}}}$. Note that $G(v,\theta_v=0)$ is an identity gate that will not appear in $QC$, and thus $QC$ has an optimal gate count ${m_D}=|\{G(v,\theta_v \neq 0)\}|$ for any $D(\overrightarrow{\alpha})$. Specifically, $QC$ has an optimal circuit depth when the implementation of $D(\overrightarrow{\alpha})$ only employs pairs of complementary gates. For example, this theorem gives us the circuit 
    in Fig.~\ref{fig2}(c) 
     to implement Fig.~\ref{fig2}(a).

\end{theorem}

    In summary, we provide a gate-count optimal circuit synthesis (that is, \textbf{Theorem}~\ref{pair-wise_synthesis}) for realizing a given diagonal unitary matrix in Eq.~\eqref{Given Diagonal}, which also enables the circuit depth optimal when all obtained  non-zero angle parameters correspond to pairs of complementary gates. Furthermore, in the following we  consider how to optimize the depth of any other types of MCZR circuits . 

\section{Depth optimization of MCZR circuits}\label{depth optimization}
 Since all MCZR gates are diagonal and commutative, the task of optimizing the depth of any given MCZR circuit is equivalent to rearranging all its gates into as few disjoint layers as possible. In this section, we propose a gate-exchange strategy together with a flexible algorithm for effectively reducing the circuit depth.

 \subsection{A gate-exchange strategy for optimizing the circuit depth}\label{pre-processing strategy} 

First of all, we present a simple but useful strategy in \textbf{Lemma}~\ref{theorem_complementary} that can reduce (or retain) the depth of any MCZR circuit.

\begin{lemma}  \label{theorem_complementary}
For a depth-$d_1$ MCZR circuit $QC_1$ over the gate set $S=\{G(v,\theta_v)\}$, suppose that (1) a pair of complementary gates $G(v_1,\theta_{v_1})$ and $G(v_2,\theta_{v_2})$ are located in two different layers of $QC_1$, and (2) the gate $G(v_1,\theta_{v_1})$ and a subset of gates $\{G(v',\theta_{v'})\} \subset S$ are located in the same layer of $QC_1$. Then, the exchange of $\{G(v',\theta_{v'})\}$ and $G(v_2,\theta_{v_2})$ in $QC_1$ would arrange $G(v_1,\theta_{v_1})$ and $G(v_2,\theta_{v_2})$ into one layer, leading to a new depth-$d_2$ circuit $QC_2$  with $d_2\le d_1$.
\end{lemma}

We give an intuitive explanation of \textbf{Lemma}~\ref{theorem_complementary}. In the original depth-$d_1$ circuit $QC_1$, suppose the gate $G(v_1,\theta_{v_1})$ and gates in $\{G(v',\theta_{v'})\}$ are located in a layer indexed by $L_1$, while the gate $G(v_2,\theta_{v_2})$ is located in another layer indexed $L_2$. 
Then the exchange of $G(v_2,\theta_{v_2})$ and $\{G(v',\theta_{v'})\}$ arranges the former and the latter into the layer $L_1$ and $L_2$, respectively. Since the gate $G(v_2,\theta_{v_2})$ alone acts on more qubits than any gate in $\{G(v',\theta_{v'})\}$ does, such a gate-exchange operation  would lead to two possible situations about the resultant circuit $QC_2$: (1) $QC_2$ has the same depth $d_1$ as $QC_1$, or (2) some (or all) of the gates in $\{G(v',\theta_{v'})\}$ and the gates adjacent to layer $L_2$ can be merged into the same layer, thus causing a depth reduction  over $QC_1$.

Based on \textbf{Lemma}~\ref{theorem_complementary}, we can derive a two-step framework for achieving a depth-optimal MCZR circuit as described in \textbf{Lemma}~\ref{optimal_any_depth}.

\begin{lemma}\label{optimal_any_depth}
In principle, the optimal circuit depth $d_{opt}$ of the MCZR circuits constructed from a given gate set  $S=\{G(v,\theta_v)\}$ with $|S|=m$ can be achieved by two steps: (1) arrange all $d_1$  pairs of complementary gates in $S$ into a depth-$d_1$ configuration, and (2) find a depth-optimal circuit over the other $r=(m-{2d_1})$ gates. Then $d_{opt}$ is equal to the total depth of these two parts.
\end{lemma}

A special case of \textbf{Lemma}~\ref{optimal_any_depth} is \textbf{Theorem}~\ref{theorem_optimal_depth}, such that $m=2d_1$ gives us ${d_{opt}={d_1}}$. In general, we can accomplish the second step of \textbf{Lemma}~\ref{optimal_any_depth}
by comparing at most $r!$ different layer configurations and find the depth-optimal circuit over a given gate set $S$. However, for $S$ with a moderate value $r$, the number of all possible layer configurations can be quite large and thus the optimal depth is usually hard to determine. To deal with such complicated cases, in the following we further propose a flexible iterative algorithm for optimizing the depth of a circuit with no complementary gates.

\subsection{A flexible iterative depth-optimization algorithm}\label{strategy1}

In this section, we propose an iterative algorithm denoted \textbf{Algorithm} \ref{greedy_layer} for optimizing the depth of MCZR circuits with no complementary gates, and reveal its flexibility with a use case.

\SetKwIF{If}{ElseIf}{Else}{if}{then}{else if}{else}{end\ if}
\SetKwFor{For}{for}{do}{end\ for}
\SetKwFor{ForEach}{foreach}{do}{}
\SetKwFor{While}{while}{do}{end\ while}

\begin{algorithm}[!htpb]
    \caption{An iterative depth-optimization algorithm for MCZR circuits.}
    \label{greedy_layer}
    \LinesNumbered
    \SetAlgoNoLine
    \SetKwFunction{funcGLF}{Greedy\_Layer\_Formation}
    \SetKwFunction{funcGNG}{Generate\_New\_GateSeq}
    \KwIn{A depth-$d$ MCZR circuit $QC$ with its constituent gates located from left to right
    as a sequence 
$SEQ=[act_k:k=1,2,\ldots,m]$, with $act_k$ being the qubit set of the $k$th gate; an iteration number $iter\geq 1$. 
     }
    \KwOut{A circuit $QC_{opt}$ over gates in $SEQ$ with a layer configuration $R=\{L_i: i=1,2,\ldots,d_{opt}\}$ such that ${d_{opt}}\leq d$.}

    \BlankLine
    
    \textbf{main program:}
    
    Calculate the circuit depth lower bound $LB$ for $SEQ$ by Eq.~\eqref{depth_lower_bound}.

    $[R^{(1)},d^{(1)}]=\funcGLF{SEQ}$; $t\leftarrow 1$;

\If(\tcp*[h]{Perform iterative layer  formation. }){    $d^{(1)}>LB$ $\&\&$ $iter \ge 2$ }{
        
      \For{$t \leftarrow 2$ \KwTo $iter$}{ 
      
      $SEQ^{(t)}=\funcGNG{{\text{$R^{(t-1)}$}}}$;
      
      $[R^{(t)},d^{(t)}]=\funcGLF{{\text{$SEQ^{(t)}$}}}$;

        \If(){$d^{(t)}==LB$}{
            break;
        }
    }
}    
    $d_{opt}\leftarrow d^{(p)}=min\{d^{(q)}:q\in [t]\}$; $R\leftarrow R^{(p)}$;
   
    \KwRet{$[R,d_{opt}].$}

   \BlankLine
    \SetKwProg{fn}{function}{:}{}
    \fn{\funcGLF{SEQ}}{
  $i\leftarrow0$; 
    
    \While{$|SEQ| \neq 0$}{
        $i\leftarrow i+1$;    
        $c\leftarrow 0$; 
        $L_i\leftarrow \varnothing$; 
        $remove\_set\leftarrow \varnothing$; 
        
        \For(\tcp*[h]{Greedily form the layer $L_i$. }){$k\leftarrow1$ \KwTo $|SEQ|$} {
            \If{$L_i$ and $SEQ[k]$ have no  integers in common}{ 
                $c\leftarrow c+1$;  
                $L_i[c]\leftarrow SEQ[k]$;
                $remove\_set[c]\leftarrow k$;   
            }
            
        }
        \textbf{Delete} $SEQ[remove\_set]$;
    }
     $d\leftarrow i$;
     
      \KwRet{ $[R=\{L_1, L_2, \ldots,L_d\},d]$ }.
      
    \Indm  
    \textbf{end function}
}

    \BlankLine
     \fn{\funcGNG{ $R=\{{L_i}=[act^i_1,act^i_2,\ldots,act^i_{|L_i|}]:i=1,2,\ldots,d  \}$  } }
{ $SEQ=[act^1_1,act^2_1,\ldots,act^d_1,act^1_2,act^2_2,\ldots,act^d_2,\ldots,act^p_{|L_p|}]$ with the layer index $p$ such that $|L_p|=max\{|L_i|:i\in [d]\}$;

 \KwRet{ $SEQ$ }.
      
    \Indm  
    \textbf{end function}
}

\end{algorithm}

The input of \textbf{Algorithm} \ref{greedy_layer} includes: a given MCZR circuit $QC$ with its constituent gates located from left to right as a sequence $SEQ=[act_k:k=1,2,\ldots,m]$, with $act_k$ being the set of qubits
     acted upon by the $k$th gate, and an iteration number $iter\in \mathbb{N}^+$. The output is a circuit over gates in $SEQ$ that has a depth smaller than or equal to that of $QC$.
     Notice that two subroutine functions \funcGLF and \funcGNG are introduced here: the former receives a gate sequence $SEQ$ and can arrange as many disjoint gates in $SEQ$ into each layer as possible to form a circuit layer configuration $R$, while the latter can generate a new gate sequence $SEQ$ from a given circuit $R=\{L_i:i=1,2,\ldots,d\}$ by extracting and regrouping gates in original layers $L_i$. Since the application of our greedy layer formation procedure on different sequences over a given MCZR gate set may result in distinct circuits, we will iteratively use these two functions in our \textbf{main program} to seek circuits with the shortest possible depth as follows. 
     
     First, since two gates that act on the same qubit must be located in different layers of a circuit,   a depth lower bound $LB$ on all possible circuits constructed from the input gate set $SEQ$ can be derived as:
    \begin{equation}
        \label{depth_lower_bound}
        LB(SEQ)=max\{\textsc{Count}(j,SEQ):j\in[n]\},
    \end{equation}
where $\textsc{Count}(j,SEQ)$ indicates the number of integer $j$ appeared in $SEQ$. Second, we apply the function \funcGLF to the input gate sequence $SEQ$ and obtain a new depth-$d^{(1)}$ circuit with layer configuration $R^{(1)}$ such that $d^{(1)}\leq d$. Third, if $d^{(1)}>LB$ and   $iter\geq 2$, we can further iteratively generate a new gate sequence $SEQ^{(t)}$ from the previous circuit $R^{(t-1)}$ via \funcGNG, followed by  applying \funcGLF to obtain a new circuit $R^{(t)}$ of depth $d^{(t)}$ in each loop $t\geq 2$. In this process, we can terminate the loop when getting the optimal depth as $d^{(t)}=LB$. Finally, we choose the circuit with shortest depth among all constructed $\{R^{(t)}\}$ above as our output depth-optimized circuit $R=\{L_1,L_2,\ldots,L_{d_{opt}}\}$.
As a result, \textbf{Algorithm} \ref{greedy_layer} ensures that: (1) $d_{opt}\leq d^{(1)}\leq d$, and (2) ${d_{opt_2}}\leq{d_{opt_1}}$ for two iteration numbers  ${iter_2}\geq {iter_1}$. Therefore, our \textbf{Algorithm} \ref{greedy_layer} controlled by an iteration number $iter$ is a flexible depth-optimization algorithm
by considering the relation between the reduced depth and optimization time cost.

\begin{figure}[htp]
    \centering
\includegraphics[width=0.47\textwidth]{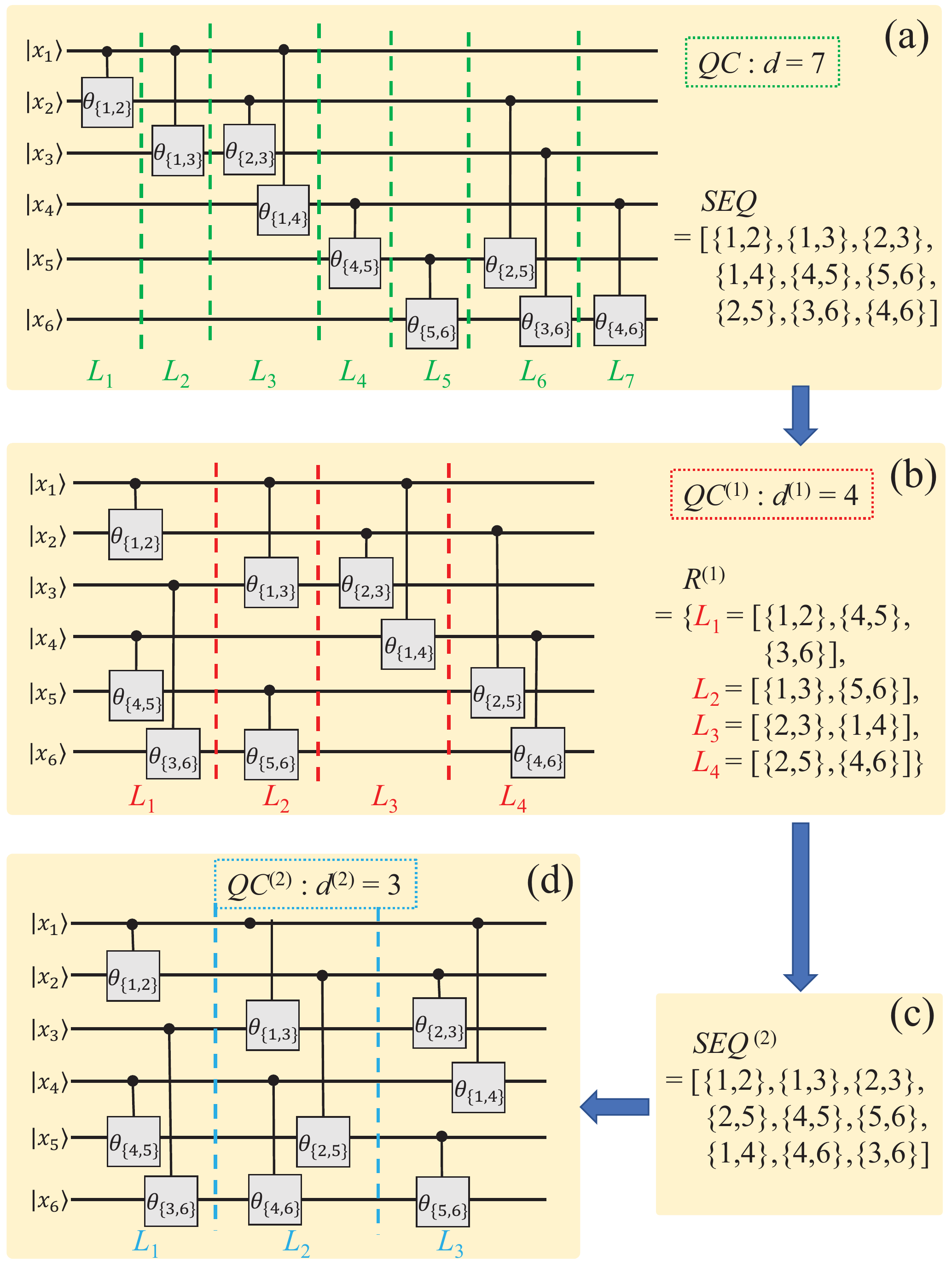}
\caption{An example to demonstrate \textbf{Algorithm}~\ref{greedy_layer} with $iter=2$. (a) A given 6-qubit MCZR circuit $QC$ of depth $d=7$, with its 9 two-qubit gates $CZ(\theta_{\{i,j\}})$  being separated by green dashed lines as $\{L_1,L_2,\ldots, L_7\}$ and in a sequence $SEQ=[\{1,2\},\{1,3\},\ldots,\{3,6\},\{4,6\}]$. The circuit depth lower bound for $SEQ$ is $LB=3$ by Eq.~\eqref{depth_lower_bound}. Then, we apply the function {\protect \funcGLF} to (a) and obtain a circuit $QC^{(1)}$ of depth $d^{(1)}=4$ as shown in (b), where its four gate layers are separated by red dashed lines as $R^{(1)}=\{L_1,L_2,L_3,L_4\}$ and Eq.~\eqref{four_layers}. Due to $d^{(1)}>LB$ and $iter=2$, next we apply the function {\protect\funcGNG} to $R^{(1)}$ and generate a new gate sequence $SEQ^{(2)}$ in (c). Once again, we apply {\protect\funcGLF} to (c) and  obtain a new circuit $QC^{(2)}$ of depth $d^{(2)}=3$ in (d), achieving the optimal circuit depth $LB$. }
    \label{figALG1}
\end{figure}

A demonstrative example of \textbf{Algorithm} \ref{greedy_layer} is shown in Fig.~\ref{figALG1}. The gate sequence for the 6-qubit and depth-7 circuit $QC$ consisting of 9 two-qubit $CZ(\theta)$ gates as shown in Fig.~\ref{figALG1}(a) is 
\begin{align}
\label{six_sequence}
    SEQ=[&\{1,2\},\{1,3\},\{2,3\},\{1,4\},\{4,5\},
    \nonumber
    \\ &\{5,6\},\{2,5\},\{3,6\},\{4,6\}],
\end{align}
and we apply \textbf{Algorithm} \ref{greedy_layer} with $iter=2$ to achieve a depth-optimized circuit as follows:  
\begin{enumerate}
\item[(1)] First, we calculate the depth lower bound on circuits for $SEQ$ by Eq.~\eqref{depth_lower_bound} as $LB=3$.
\item[(2)] Second, we apply  \funcGLF to $SEQ$ in Eq.~\eqref{six_sequence} and obtain a new circuit $QC^{(1)}$ of depth ${d^{(1)}}=4$ as shown in Fig.~\ref{figALG1}(b), which has a layer configuration $R^{(1)}=\{L_1,L_2,L_3,L_4\}$ with 
\begin{align}
\label{four_layers}
~~~~ &L_1=[act^1_1=\{1,2\},act^1_2=\{4,5\}, act^1_3=\{3,6\}],\nonumber \\ 
&L_2=[act^2_1=\{1,3\},act^2_2=\{5,6\}],\nonumber \\   
& L_3=[act^3_1=\{2,3\},act^3_2=\{1,4\}],\nonumber \\  
&L_4=[act^4_1=\{2,5\},act^4_2=\{4,6\}].
\end{align}
Intuitively, the comparison between the circuit $QC$ in Fig.~\ref{figALG1}(a) and $QC^{(1)}$ in Fig.~\ref{figALG1}(b) reveals that the working principle of our function \funcGLF is to move the gates in the right column of original circuit to fill the vacancies in the left column as much as possible, thus causing a circuit depth reduction.  
\item[(3)] Third, we apply \funcGNG to $R^{(1)}$ in Eq.~\eqref{four_layers} due to the condition $d^{(1)}>LB$ and $iter>1$, and generate a new gate sequence $SEQ^{(2)}$  
shown in Fig.~\ref{figALG1}(c) as 
\begin{align}
\label{new_sequence}
    SEQ^{(2)}=[&\{1,2\},\{1,3\},\{2,3\},\{2,5\},\{4,5\},\nonumber
    \\ &\{5,6\},\{1,4\},\{4,6\},\{3,6\}].
\end{align}
\item[(4)] Finally, we apply \funcGLF again to Eq.~\eqref{new_sequence} and obtain a new layer configuration $\{L_1,L_2,L_3\}$, that is,  the circuit $QC^{(2)}$ of depth ${d^{(2)}}=3$  as shown in Fig.~\ref{figALG1}(d). 
\end{enumerate}

Note that if we apply \textbf{Algorithm}~\ref{greedy_layer} with only $iter=1$ to $SEQ$ in Fig.~\ref{figALG1}(a), the resultant depth-optimized circuit would be just $QC^{(1)}$ in Fig.~\ref{figALG1}(b). This simple example implies that, if we apply \funcGLF to more distinct gate sequences generated from \funcGNG, the more significant depth reduction over the original circuit is likely to occur at the expense of more optimization time. More practical cases of \textbf{Algorithm}~\ref{greedy_layer} will be demonstrated 
in Section~\ref{applications}.

\section{Experimental evaluation}\label{applications}

To further evaluate the performances of the proposed synthesis and optimization methods, here we refine them into two explicit workflows and consider their applications to two typical use cases in quantum computing. All experiments are performed with MATLAB 2022a on an  Intel Core i5-12500 CPU operating at 3.00 GHz frequency and 16GB of RAM.

\subsection{Workflow of our synthesis and optimization methods}
\label{workflow_methods}

For convenience, here we summarize the main results in Secs.~\ref{synthesis algorithm} and \ref{depth optimization} into the workflow to fulfill two types of tasks as follows:

\textbf{Task 1}: How to construct a gate-count optimal MCZR circuit followed by further depth-optimization for implementing a given diagonal unitary matrix in Eq.~\eqref{Given Diagonal}? 

\textbf{Workflow 1}: First, we synthesize a gate-count optimal MCZR circuit according to \textbf{Theorem}~\ref{pair-wise_synthesis} with $m$ gates, which includes two parts: (i) $d_1$ layers of complementary gates denoted $QC_1$, and (ii) the other $(m-2{d_1})$ gates. Second, we apply \textbf{Algorithm}~\ref{greedy_layer} with a specified parameter $iter$ to optimize the part (ii) into a depth-$d_2$ circuit $QC_2$. Finally, the overall output circuit is $QC={QC_1}\circ {QC_2}$ of depth ${d_1}+{d_2}$.

\textbf{Task 2}: How to optimize the circuit depth of a given MCZR circuit $QC$ over the gate set  $S=\{G(v,\theta_v)\}$ with $|S|=m$? 

\textbf{Workflow 2}:  First, we perform the gate-exchange operation to $QC$ according to \textbf{Lemma}~\ref{theorem_complementary}, which 
arranges all $d_1$  pairs of complementary gates in $S$ into a depth-$d_1$ circuit  denoted $QC_1$. Second, we apply \textbf{Algorithm}~\ref{greedy_layer} to the other $(m-{2d_1})$ gates and obtain a circuit $QC_2$ of depth $d_2$. Finally, putting these results together gives a depth-optimized circuit $QC_{opt}={QC_1}\circ {QC_2}$ of depth ${d_1}+{d_2}$.

In the following, we demonstrate the utility of above workflows for two  practical quantum computing tasks: (1) constructing diagonal Hermitian quantum operators, and (2) optimizing the depth of QAOA circuits.

\subsection{Diagonal Hermitian quantum operators}
\label{diagonal Hermitian}
 We use $D_H^{(n)}$ to denote an 
 $n$-qubit diagonal Hermitian quantum operator with its diagonal elements being $\pm 1$, and there are totally $2^{2^n-1}$ different such operators since  $D_H^{(n)}$ and $- D_H^{(n)}$ are essentially equivalent. Note that operators of this type act as the oracle operator or fixed operator in the well-known Deutsch-Jozsa algorithm \cite{Deutsch1992,Collins1998}, Grover's algorithm \cite{grover1996fast} and some recent  algorithms showing quantum advantage for string learning and identification \cite{XU2023quantum,Li2022,Huang2022}. Therefore, an efficient construction of $D_H^{(n)}$ over MCZR gates would facilitate the  implementation of relevant quantum algorithms on specific devices~\cite{song2017continuous,hill2021realization}.   
 
Prior work \cite{houshmand2014decomposition} has  revealed that $D_H^{(n)}$ can be synthesized by at most $2^n-1$ multiple-controlled Pauli $Z$ gates,  that is, MCZR gates with a fixed angle parameter $\pi$,  based on a binary representation and solving linear equations over the binary field ${\mathbb{F}}_2$. As comparison, here we apply our synthesis and optimization methods to construct circuits for realizing such operators, and to be more specific, our strategies include: pair-wise synthesis method  in \textbf{Theorem}~\ref{pair-wise_synthesis}  ($app01$), our \textbf{Workflow 1} in Sec.~\ref{workflow_methods} with $iter=1$ ($app02$), $iter=5$ ($app03$), and $iter=20$ ($app04$), respectively. We perform experiments on all 8, 128, 32768 diagonal Hermitian operators $D_H^{(n)}$ for $n=2,3,4$, respectively, as well as 100 randomly selected ones for each $ 5 \le n \le 12 $, and compare our results with the previous work.  
Due to the uniqueness property , our constructed circuits have the same MCZR gate set as that from Ref.~\cite{houshmand2014decomposition}, and therefore we mainly illustrate our circuit depth reduction. The detailed experimental results are presented in Fig.~\ref{fig_Hermitian}.

\begin{figure*}
    \centering
    \includegraphics[width=0.92\textwidth]{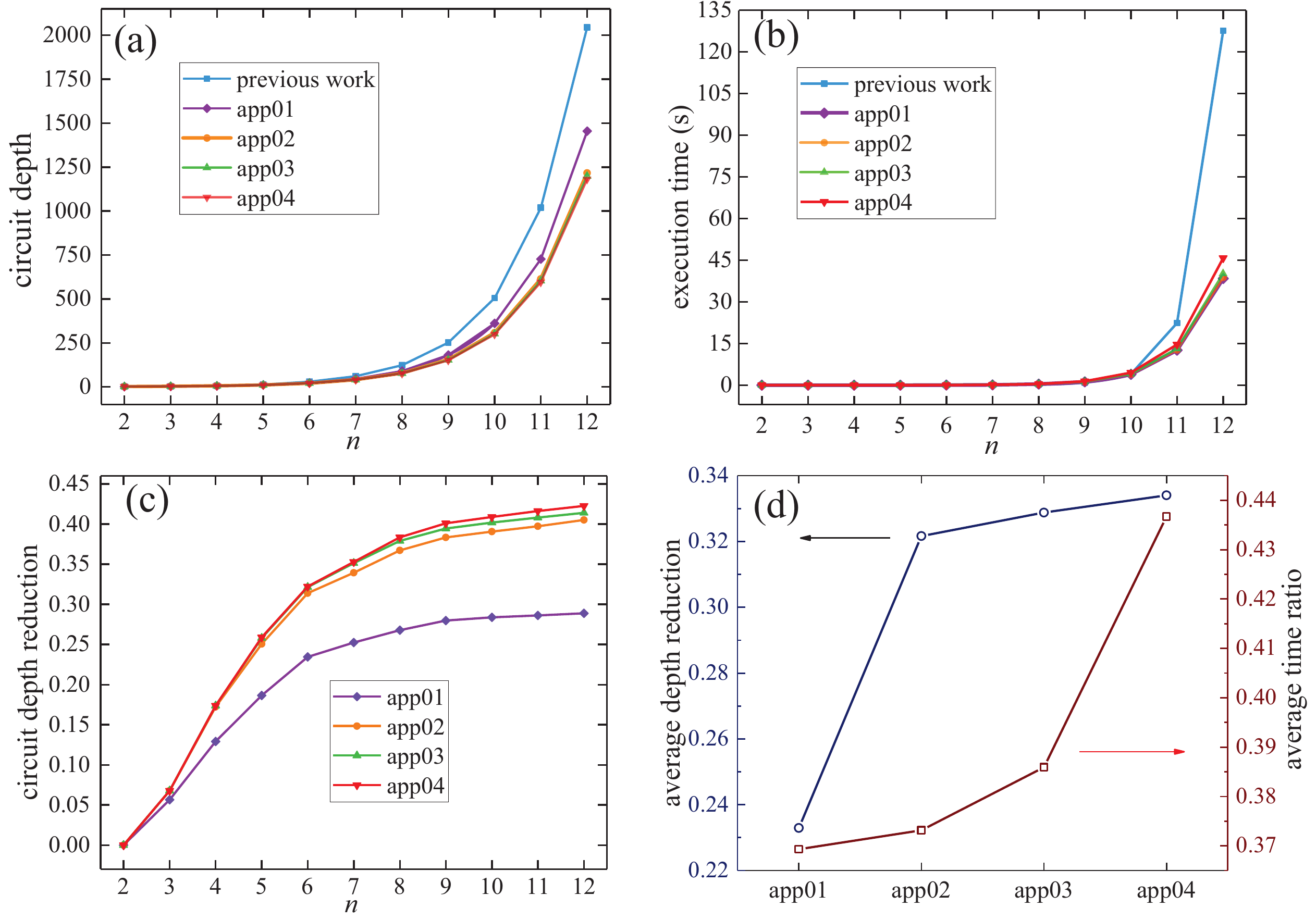}
    \caption{
    Experimental results of constructing diagonal Hermitian operators 
     with size $n\in [2,12]$ by applying a previous method \cite{houshmand2014decomposition}, our circuit synthesis method in \textbf{Theorem}~\ref{pair-wise_synthesis}  ($app01$), and our \textbf{Workflow 1} with $iter=1 (app02), 5(app03)$, and 20($app04$)
    , respectively. (a) The blue, purple, orange, green, and red curves indicate the average depth of circuits obtained from previous work and $app01$ to $app04$ for each $n$. Accordingly, the execution time and circuit depth reduction over the previous work as a function of $n$ on average are respectively recorded in
    (b) and (c), indicating that our four strategies can achieve both a reduced circuit depth and less execution time compared to previous work. Notably, all our strategies can have a more significant depth reduction for large-size $n$, and the effectiveness of our depth-optimization \textbf{Algorithm} \ref{greedy_layer} can be reflected by comparing $app02$-$app04$ with $app01$. 
    (d) As an overall performance evaluation, the average depth reduction and time ratio of our four strategies over the previous work for the entire set of instances are displayed in dark blue and dark red lines, respectively, such that on average we can achieve a 33.40\% depth reduction with only 43.67\% time by $app04$.}
    
    \label{fig_Hermitian}
\end{figure*}

In Fig.~\ref{fig_Hermitian} (a), we present 
the average circuit depth of $n$-qubit MCZR circuits ($n\in [2,12]$) constructed from the previous work \cite{houshmand2014decomposition}, our four strategies $app01$, $app02$, $app03$, and $app04$ by the blue, purple, orange, green, and red curve, respectively. Accordingly, the average execution time of constructing a circuit of size $n$ by these strategies are recorded in Fig.~\ref{fig_Hermitian} (b). Typically, the time growth of our sole circuit synthesis algorithm $app01$ as a function of $n$ agrees well with the total time complexity of calculating Eq.~\eqref{theorem_angle}, that is, $\propto n3^n$. As comparison, the time of previous work \cite{houshmand2014decomposition} increases more drastically wih $n$, since its  most time-consuming procedure for solving linear equations over ${\mathbb{F}}_2$ to determine whether each MCZR gate exists or not would require time scaling roughly as $O(N^3)=O(8^n)$. 
It is worth noting that all our four strategies have both a reduced circuit depth and less execution time over the previous work. 
In Fig.~\ref{fig_Hermitian} (c), the  circuit depth reduction curve for each of our strategies shows an explicit upward trend as the circuit size $n$ increases, which can achieve as high as 28.88\%, 40.51\%, 41.40\%, 42.27\% for constructing a circuit of $n=12$ on average in time 38.40s, 38.79s, 40.16s, and 45.78s, respectively. Also, the usefulness of \textbf{Algorithm} \ref{greedy_layer} is reflected by observing that $app02$  can achieve a 11.57\% smaller depth over the sole synthesis algorithm $app01$ at the expense of only 
 1.03\% more time for circuits of $n=12$, while $app03$ and $app04$ give us shorter and shorter depths as $iter$ increases.
Finally, in Fig.~\ref{fig_Hermitian} (d) we evaluate the overall average performances of our strategies $app01$, $app02$, $app03$, and $app04$ for all involved circuit instances with $n\in[2,12]$, including the average depth reduction of 23.29\%, 32.16\%, 32.88\%, and 33.40\%, 
and the average time ratio of 36.93\%, 37.31\%, 38.59\%, and 43.67\% with respect to the previous work, respectively.
It seems that for such circuit instances, the average depth-optimization trend would rise slowly as the iteration number $iter$ in
\textbf{Workflow 1} increases.

\begin{figure*}
    \centering
  \includegraphics[width=0.92\textwidth]{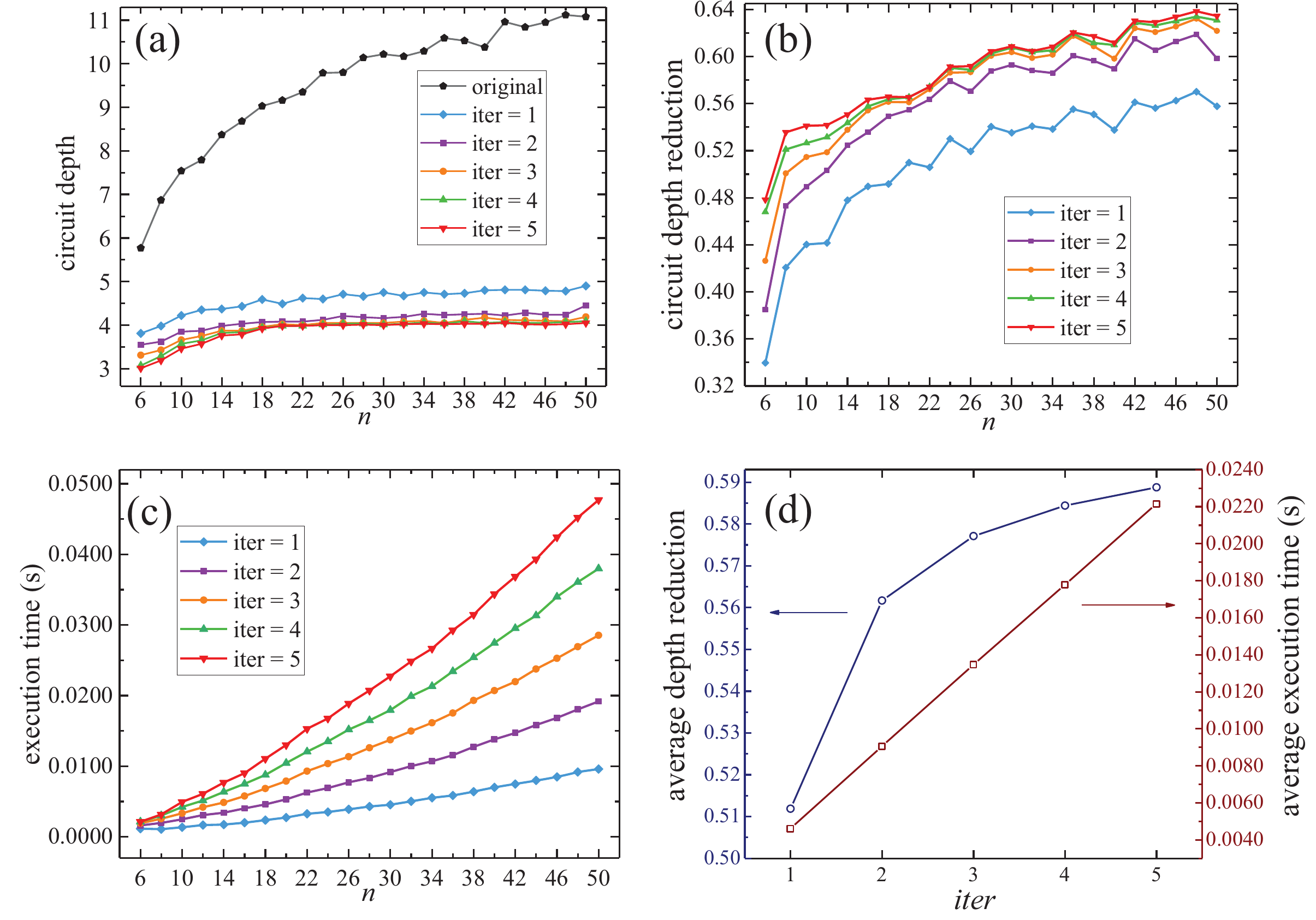}
    \caption{
    Experimental results of optimizing the depth of phase-separation parts in 100 randomly selected $n$-qubit  QAOA circuits with even $n\in [6,50]$ by applying \textbf{Algorithm}~\ref{greedy_layer} with $iter=1,2,3,4,5$, respectively. (a) The black, blue, purple, orange, green, and red curves indicate the average circuit depth of original 100 random $n$-qubit instances as well as optimized ones with $iter=1$ to 5, respectively. Accordingly, the circuit depth reduction and execution time as a function of $n$ on average are respectively recorded in
    (b) and (c), both of which show an upward trend on the whole. Note that the results for $iter=1$ are equivalent to the previous min-layer formation method aimed at optimizing QAOA circuits \cite{alam2020efficient}, while as comparison our \textbf{Algorithm}~\ref{greedy_layer} is more flexible and useful since it can achieve a more significant circuit depth reduction 
    by adjusting the parameter $iter$ at the cost of more execution time. 
    (d) As an overall performance evaluation, the average depth reduction and execution time for all 2300 circuit instances with different $iter$ are displayed in dark blue and dark red, respectively, where the time cost shows a nearly-linear growth when increasing $iter$.}
    \label{figQAOA}
\end{figure*}

In summary, here we demonstrate our \textbf{Workflow 1} for  synthesizing and optimizing MCZR circuits by taking diagonal Hermitian operators as an  example, which can show substantial improvement over the previous work in terms of both circuit depth and execution time.
In addition, our results empirically validate that a shorter circuit depth is likely to be achieved by increasing the iteration number $iter$ in \textbf{Algorithm} \ref{greedy_layer} with more time (see Fig.~\ref{fig_Hermitian}.(d)). In the following, we focus on another example to highlight the flexibility of  \textbf{Algorithm}~\ref{greedy_layer} for realizing controllable depth optimization.

\subsection{Phase-separation part in QAOA circuit}

 Quantum Approximate Optimization Algorithm (QAOA) is a well-known hybrid quantum-classical  algorithm designed to solve combinatorial
optimization problems. A typical stage of the QAOA circuit for the MaxCut problem consists of three parts: a layer of Hadamard gates, a phase-separation part consisting of $CZ(\theta)$ gates, and a layer of $R_x$ rotation gates. Here we focus on reducing the depth of the middle part in  $n$-qubit MaxCut-QAOA circuits of 3-regular graphs
\cite{alam2020efficient} by using our \textbf{Workflow 2} in Sec.~\ref{workflow_methods}, which is thus   \textbf{Algorithm}~\ref{greedy_layer} for  $n\ge 6$. 

   To our knowledge, 
prior work ~\cite{alam2020efficient} has used a so-called min-layer formation (MLF) procedure  for reducing the 
number of  $CZ(\theta)$ gate layers in QAOA circuits, which is exactly 
a particular case of our \textbf{Algorithm}~\ref{greedy_layer} with  the iteration number taken as $iter=1$. 
For comparison, here we  apply   \textbf{Algorithm}~\ref{greedy_layer} with $iter=1,2,3,4,5$ to optimize such phase-separation part consisting of two-qubit $CZ(\theta)$ gates in 
OAOA circuits, respectively. 
According to the definition of 3-regular graphs such that every vertex is connected to  three other vertices, the circuit depth lower bound in Eq.~\eqref{depth_lower_bound} is determined to be 3 for any circuit instance input to  \textbf{Algorithm}~\ref{greedy_layer}. 
As an example, the depth optimization of a 6-qubit phase-separation circuit $QC$ of depth 7 by taking $iter=2$ has been presented in Fig.~\ref{figALG1}.
More broadly, here we pick the $n$-qubit circuit instances corresponding to $n$-node 3-regular graphs with $n$  being an even number in the range of 6 to 50, and for each size $n$ we  randomly
pick 100 graphs. Thus, a total of $23 \times 100=2300$  
MaxCut-QAOA circuit instances have been used for the evaluation. 
The experimental results are  presented in Fig.~\ref{figQAOA}.

The average circuit depth of 100 original randomly selected $n$-qubit QAOA circuits for $n\in [6,50]$ is shown as the black curve in Fig.~\ref{figQAOA} (a), where the blue, purple, orange, green, and red curves indicate the optimized circuit depth obtained from performing \textbf{Algorithm}~\ref{greedy_layer} with
$iter=1$ (that is, MLF procedure in Ref.~\cite{alam2020efficient}) as well as $iter=2,3,4,5$, respectively. Specifically, 
  the optimized circuit depths as indicated by the red line in Fig.~\ref{figQAOA}(a) with $iter =5$ grows quite slowly and ranges from 3.00 to 4.05 for $n\in [6,50]$. Accordingly, Figs.~\ref{figQAOA} (b) and ~\ref{figQAOA}(c) show the circuit depth reduction and execution time for each instance with size $n$ on average, respectively. In particular,  the depth-reduction curve for each setting $iter$ is growing overall as the circuit size $n$ increases,  and can achieve as high as 63.45\% for $n=50$ in time less than 0.05s when adopting $iter=5$. Furthermore, Fig.~\ref{figQAOA} (d) shows the overall  performance of \textbf{Algorithm}~\ref{greedy_layer} with $iter = 1,2,3,4,5$ on all 2300 circuit instances, where on average we can 
achieve a depth reduction of 51.19\%, 56.17\%, 57.71 \%, 58.44\%, and 58.88\% over one original randomly selected QAOA circuit instance by using time of 0.0046s, 0.0090s, 0.0135s, 0.0178s and 0.0222s for  each $iter\in [1,5]$, respectively. 
Notably, the average execution time scales nearly linearly as $iter$ increases from 1 to 5, and the average depth obtained from $iter=5$ is 15.55\% smaller than that from $iter=1$ at the expense of 4.81X increase in time.
Once again, these results reflect the flexibility of  \textbf{Algorithm}~\ref{greedy_layer} as it can achieve a shorter circuit depth at the expense of more execution time. 
Therefore, for dealing with such QAOA-circuit case one can take \textbf{Algorithm} \ref{greedy_layer}  
with gradually increasing the iteration number $iter$ to seek the best possible results.
 
Finally, we point out the expense of depth-optimization time overhead is especially  worthwhile in the use-case of QAOA since the obtained circuit needs to be executed on the quantum hardware many times for solving the MaxCut problem, and thus a shorter circuit depth obtained from the precedent optimization procedure could save a large amount of time in the subsequent process of running the QAOA circuit. 
As a result, our depth-optimized circuits might be executed on the scalable quantum processor with non-local connectivity~\cite{Bluvstein2022}, or can act as a better starting point for possible further circuit compilation if needed \cite{alam2020efficient}. 

\section{Discussion and Conclusion}
\label{Conclusion}
In this study, we present a systematic study of quantum circuits over multiple-control $Z$-rotation gates with continuous parameters. Based on an established polynomial representation, we derive a gate-count optimal synthesis of such circuits for implementing any diagonal unitary matrix, which also enables the circuit depth optimal for specific MCZR circuits. Furthermore, we propose practical optimization strategies for reducing the circuit depth of any given MCZR circuit, which can show substantial performance improvement over prior works for typical examples in quantum computing. Compared to the conventional study of implementing diagonal unitary operators over the single- and two-qubit gate set \cite{markov2004,2014welch,zhang2022automatic}, here we provide an alternative scheme by utilizing a multiqubit gate set as the computational primitives, which would match the quantum experimental progress in certain directions, such as neutral atoms \cite{levine2019parallel} and superconducting systems \cite{song2017continuous,2020PhysRevApplied.14.014072}. 
In addition, note that above techniques are raised for dealing with general cases, we point out there may also exist other useful ideas aimed at special-case circuits. For example, particular quantum graph states \cite{bravyi2018quantum} or hypergraph states \cite{PRXQuantum020333} can be prepared with linearly many MCZR gates and constant depth by observing their underlying lattice graphs. Readers of interest could explore more about such specific cases.

Although this paper mainly focuses on quantum circuits over MCZR gates, it may enlighten the research on  other types of circuits as well.  First, the circuit-polynomial correspondence put forward to characterize MCZR circuits 
extends the concept of \textit{phase polynomial} representation \cite{2018nam}, again implying that an appropriate representation could facilitate circuit synthesis and/or optimization. Second, the depth-optimization strategies introduced in Section~\ref{depth optimization} are actually suitable for any quantum circuit over commuting gates, such as IQP (instantaneous quantum
polynomial-time) circuits used to demonstrate quantum advantage \cite{Bremner2017}. 
Finally, this study sheds light on implementing diagonal unitary operators over other available gate sets, such as  the multiply-controlled
Toffoli gates acting on fewer qubits by considering gate simulation \cite{barenco1995elementary}. Therefore, we would like to investigate these interesting topics in the future work.

\begin{acknowledgments}
This work was supported by the National Natural Science Foundation of China (Grant Nos. 62102464, 62272492, 61772565),
 the Guangdong Basic and Applied Basic Research Foundation
(Grant No. 2020B1515020050), and Project funded by
China Postdoctoral Science Foundation (Grant Nos. 2020M683049,
2021T140761).  We appreciate Dr. Li Zhang from South China 
Normal University for useful discussions on the data analysis.
\end{acknowledgments}


\bibliography{apssamp}

\end{document}